\documentclass[format=acmsmall]{acmart}

\usepackage{bm}
\usepackage[subrefformat=parens]{subcaption}

\AtBeginDocument{%
  \providecommand\BibTeX{{%
    \normalfont B\kern-0.5em{\scshape i\kern-0.25em b}\kern-0.8em\TeX}}}

\setcopyright{acmcopyright}
\copyrightyear{2019}
\acmYear{2019}
\acmDOI{0}

\acmJournal{TOMS}
\acmVolume{0}
\acmNumber{0}
\acmArticle{0}
\acmMonth{1}

\begin{document}

\title{Again, random numbers fall mainly in the planes: 
  xorshift128+ generators}

\author{Hiroshi Haramoto}
\email{haramoto@ehime-u.ac.jp}
\affiliation{%
  \institution{Faculty of Education, Ehime University}
  \streetaddress{3 Bunkyocho}
  \city{Matsuyama}
  \state{Ehime}
  \postcode{790-8577}
}

\author{Makoto Matsumoto}
\email{m-mat@math.sci.hiroshima-u.ac.jp}
\affiliation{%
  \institution{Graduate School of Science, Hiroshima University}
  \streetaddress{1-3-1 Kagamiyama}
  \city{Higashi-Hiroshima}
  \state{Hiroshima}
  \postcode{739-8526}
}








\renewcommand{\shortauthors}{Haramoto and Matsumoto}

\begin{abstract}
  Xorshift128+ are pseudo random number generators with eight sets
  of parameters. 
  Some of them are standard generators in many platforms, 
  such as JavaScript V8 Engine. 
  We show that in the 3D plots generated by this method, 
  points concentrate on planes, ruining the randomness. 
\end{abstract}

\ccsdesc[500]{Mathematics of computing~Mathematical software}
\ccsdesc[500]{Theory of computation~Pseudorandomness and derandomization}



\keywords{pseudo random number generators, xorshift128+, discrete mathematics}

\maketitle

\section{Notation and xorshift128+}
Let $\bm{x}$ be a $64$-bit (unsigned) integer. 
Let $\mathbb{F}_2 = \{0,1\}$ denote the two-element field, 
and $\bm{x} \in \mathbb{F}_2^{64}$ is considered 
to be a $64$-dimensional row vector. 
A linear transform $\bm{x} \mapsto \bm{x}L$ is defined as the left shift 
by $1$ bit, and $\bm{x} \mapsto \bm{x}R$ is the right shift by $1$ bit.
The matrix 
$\begin{pmatrix} 0 & & & & \\ 1 & 0 & & & \\ & 1 & \ddots & & \\
  & & \ddots & \ddots & \\ & & & 1 &0 \end{pmatrix}$ is $L$ and 
$\begin{pmatrix} 0 & 1 & & & \\  & 0 & 1 & & \\ & & \ddots & \ddots & \\
  & & & \ddots & 1 \\ & & &  & 0 \end{pmatrix}$ is $R$. 
Let $I$ denote the identity matrix. 

The xorshift128+ \cite{VIGNA2017175} pseudo random number generator (PRNG) 
has $128$-bit state space. 
A state consists of two $64$-bit words $(s_i, s_{i+1})$, 
and the next state is $(s_{i+1}, s_{i+2})$, where
\begin{equation}
\label{recursion1}
s_{i+2}  = s_i (I \oplus L^a) (I \oplus R^b) \oplus s_{i+1}(I \oplus R^c).
\end{equation}
Here the notation $\oplus$ is used for bitwise xor operation, 
or equivalently addition of vectors in $\mathbb{F}_2^{64}$ and that
of $64 \times 64$ matrices $I$ and $L^a$, etc.
The output $o_i$ at the $i$-th state $(s_i, s_{i+1})$ is given by 
\[
o_i = s_i + s_{i+1} \bmod{2^{64}}, 
\]
where $+$ denotes addition of $64$-bit integers.
The $128$-bit state $(s_0, s_1)$ is the initial state. 

The consecutive three outputs from $(s_i, s_{i+1})$ is 
$x=s_i+s_{i+1} \bmod{2^{64}}$, 
$y=s_{i+1}+s_{i+2} \bmod{2^{64}}$ 
and $z=s_{i+2}+s_{i+3} \bmod{2^{64}}$, 
and we shall show some relations among $x, y$ and $z$.


\section{Approximation of xor by sum and subtraction}
Let $x$, $y$ be $n$-bit unsigned integers. 
We consider mainly $n=3$ case, and thus only $8$ possibilities exist 
for each of $x$ and $y$. Our claim is that $x \oplus y$ is 
with non negligible probability well-approximated by 
one of $x+y$, $x-y$ or $y-x$, as analyzed below.

We consider $\bm{x}=(x_1, \ldots, x_n) \in \mathbb{F}_2^n$ 
as an $n$-dimensional vector, which is also considered as an $n$-bit 
unsigned integer denoted by $x=\sum\limits_{i=1}^n x_i 2^{n-i}$. 
In this situation, we write $\bm{x} = x$.
Let $\bm{y} = (y_1, \ldots, y_n)$ be another $n$-dimensional vector.
Then $\bm{x} \oplus \bm{y}$ is the addition of $\mathbb{F}_2$ vectors.
We discuss when $\bm{x} \oplus \bm{y} = x+y$ holds, 
where the both sides are regarded as $n$-bit integers.
Note that the operation $\bm{x} \oplus \bm{y}$ is similar to $x+y$, except that
no over flow is reflected. 
Then $\bm{x} \oplus \bm{y} \leq x+y$ holds and the equality holds 
if and only if no overflow occurs. 
Equivalently, if and only if $(x_i, y_i) \in \{(0,0), (1,0), (0,1)\}$ holds
for $i=1$, $2$, $\ldots$, $n$. 
Thus, among $4^n$ possibilities of pairs $\bm{x}, \bm{y}$, 
exactly $3^n$ pairs satisfy $x+y = \bm{x} \oplus \bm{y}$. 
If plus is taken module $2^n$, there are more cases with equality, 
e.g., $\bm{x}=(1,0,\ldots,0) = \bm{y}$ is the case.
This observation is summarized as follows.

\begin{theorem}[xor equals sum]
  \label{theoremplus}
  Let $\bm{x}, \bm{y} \in \mathbb{F}_2^n$  be $n$-bit integers
  $x$, $y$, where $\bm{x} = (x_1, \ldots, x_n)$, 
  $\bm{y} = (y_1, \ldots, y_n)$, $x=\sum\limits_{i=1}^n x_i 2^{n-i}$
  and $y=\sum\limits_{i=1}^n y_i 2^{n-i}$. 
  Then $\bm{x} \oplus \bm{y} \leq x+y$ holds, and the equality holds
  if and only if $(x_i, y_i) \neq (1,1)$ holds for $i=1$,$2$,$\ldots$, $n$. 
  Among $4^n$ pairs $(\bm{x}, \bm{y})$, $3^n$ pairs satisfy the equality. 
  More pairs satisfy $\bm{x} \oplus \bm{y} = x+y \bmod{2^n}$. 
\end{theorem}

\begin{proof}
Proof follows from the previous observation. 
\end{proof}

Another observation is about when the equality 
\[
\bm{x} \oplus \bm{y} = x-y
\]
occurs. Again, we do not take modulo $2^n$ for the right hand side. 
If we compute subtraction $x-y$ in binary without borrows, 
we obtain $\bm{x} \oplus \bm{y}$. There may be borrows, 
so we have inequality
\[
\bm{x} \oplus \bm{y} \geq x-y, 
\]
with equality holds when no borrow occurs for each digit, 
or equivalently, the pair of bits $(x_i, y_i)$ lies in 
$\{(0,0), (1,0), (1,1)\}$ for each $i=1,2,\ldots,n$. 
There are $3^n$ such pair $(\bm{x}, \bm{y})$.

\begin{theorem}[xor equals subtraction]
  \label{theoremminus}
  Let $\bm{x}, \bm{y} \in \mathbb{F}_2^n$ be as in Theorem \ref{theoremplus}.
  We have inequality
  \[
  \bm{x} \oplus \bm{y} \geq x-y,
  \]
  and the equality holds if and only if $(x_i, y_i) \neq (0,1)$ 
  for each $i=1,2,\ldots, n$. There are $3^n$ such pairs.
\end{theorem}

\begin{theorem}
  \label{counttheorem}
  Let $X$ be the set of pairs 
  $\{(\bm{x}, \bm{y}) \mid \bm{x}, \bm{y} \in \mathbb{F}_2^n\}$, and put
  \begin{align*}
    A&:=\{(\bm{x}, \bm{y}) \in X \mid \bm{x}\oplus\bm{y} = x+y\} \\
    B&:=\{(\bm{x}, \bm{y}) \in X \mid \bm{x}\oplus\bm{y} = x-y\} \\
    C&:=\{(\bm{x}, \bm{y}) \in X \mid \bm{x}\oplus\bm{y} = y-x\}.
  \end{align*}
  Then, $\# X = 4^n$, $\#A=\#B=\#C=3^n$, 
  $\#(A \cap B) = \#(B \cap C) = \#(C \cap A) = 2^n$, and 
  $\#(A \cap B \cap C)=1$ hold. 
  In particular, $\#(A\cup B \cup C)=3\cdot 3^n - 3 \cdot 2^n + 1$ holds.
\end{theorem}

\begin{proof}
  We have 
  \begin{align*}
    A&=\{(\bm{x}, \bm{y}) \in X \mid {}^\forall i \,
    (x_i, y_i) \in \{(0,0), (0,1), (1,0)\} \} \\
    B&=\{(\bm{x}, \bm{y}) \in X \mid {}^\forall i \,
    (x_i, y_i) \in \{(0,0), (1,0), (1,1)\} \} \\
    C&=\{(\bm{x}, \bm{y}) \in X \mid {}^\forall i \,
    (x_i, y_i) \in \{(0,0), (0,1), (1,1)\} \}, 
  \end{align*}
  and the second series of equalities hold. Since
  $
  A \cap B = \{(\bm{x}, \bm{y}) \in X \mid {}^\forall i \,
  (x_i, y_i) \in \{(0,0), (1,0)\} \} 
  $,
  $\#(A \cap B) = 2^n$ follows, and 
  $\#(B \cap C) = \#(C \cap A) = 2^n$ is similarly proved.
  We have 
  $A \cap B \cap C = \{(\bm{x}, \bm{y}) \in X \mid 
  {}^\forall i \, (x_i, y_i) \in \{(0,0)\} \}$, and 
  see that $\#(A \cap B \cap C) = 1$. 
  The last equality follows from the standard 
  inclusion-exclusion principle.
\end{proof}

\begin{example}
  Suppose $n=3$, i.e., we consider three-bit precision. 
  Then, among $4^3=64$ pairs $(\bm{x}, \bm{y})$, 
  we showed that at least one of $\bm{x}\oplus\bm{y}=x+y, x-y, y-x$ 
  occurs for $3\cdot3^3-3\cdot2^3+1=58$ pairs, 
  which is highly plausible:
  only $64-58=6$ exceptions exist. 
  For $n=4$, among $256$ pairs, $196$ pairs satisfy one of the three relations.
\end{example}

\section{Analysis of plus in xor, aka $+$ in xorshift128+}
We consider the consecutive three outputs
\begin{align*}
  x &= s_{i}+s_{i+1} \bmod{2^{64}} \\
  y &= s_{i+1}+s_{i+2} \bmod{2^{64}} \\
  z &= s_{i+2}+s_{i+3} \bmod{2^{64}},
\end{align*}
where $s_{i+2}$ and $s_{i+3}$ are determined from 
$(s_i, s_{i+1})$ by the recursion (\ref{recursion1}). 
We analyze
\begin{align}
  \notag z&=s_{i+2}+s_{i+3} \bmod{2^{64}} \\ 
  \label{z1}&=(s_{i+1}(I\oplus R^c) 
  \oplus s_i(I\oplus L^a)(I\oplus R^b)) \\
  \label{z2}&\quad +(s_{i+2}(I\oplus R^c) 
  \oplus s_{i+1}(I\oplus L^a)(I\oplus R^b)) \bmod{2^{64}}
\end{align}

We could not give an exact analysis, but give an intuitional approximation.
Numbers $b$ and $c$ are larger than $10$, and the most significant $b$ bits
of $\bm{x}(I+R^b)$ is identical with those of $\bm{x}$. 
Thus, as far as we concentrate on the most significant several bits
(we use Theorems mainly for $n=3$ MSBs), 
we may consider $I \oplus R^b$ and $I \oplus R^c$ to be the identity
matrix $I$. 
Thus, we have an approximation
\begin{align}
  \label{z3} z &\approx (s_{i+1} \oplus s_i(I\oplus L^a)) \\
  \label{z4} &\quad + (s_{i+2} \oplus s_{i+1}(I\oplus L^a)).
\end{align}

We denote by $\approx$ when the both sides coincide up to 
the most significant $\min\{b,c\}$ bits or some specified $n$-bits, 
except that with small probability the matrices $R^b$ and/or $R^c$, 
through the carry of $+$ between (\ref{z3}) and (\ref{z4}),
may affect on the MSBs.

Now in (\ref{z3}), we have
\begin{align}
  \notag s_i(I\oplus L^a) &= s_i \oplus s_i L^a \\
  \label{s-label1} &= s_i \oplus (2^a s_i \bmod{2^{64}})
\end{align}
which is in Theorem \ref{counttheorem}, according to the cases
$A$, $B$, $C$, 
\begin{equation*}
  s_i \oplus (2^a s_i \bmod{2^{64}}) \approx
  \begin{cases}
    (1+2^a)s_i \bmod{2^{64}} & \mbox{Case $A_i$} \\
    (1-2^a)s_i \bmod{2^{64}} & \mbox{Case $B_i$} \\
    (2^a-1)s_i \bmod{2^{64}} & \mbox{Case $C_i$} \\
    \mbox{unknown} & \mbox{otherwise}
  \end{cases}
\end{equation*}
with respect to the most significant $n$ bits.
From now on we consider the most significant $n$ bits, 
with mainly $n=3$. Again in (\ref{z3}),
\begin{equation*}
  s_{i+1} \oplus s_i(I \oplus L^a) \approx
  \begin{cases}
    s_{i+1}+s_i(I\oplus L^a) & \mbox{Case $A'_i$} \\
    s_{i+1}-s_i(I\oplus L^a) & \mbox{Case $B'_i$} \\
    -s_{i+1}+s_i(I\oplus L^a) & \mbox{Case $C'_i$} \\
    \mbox{unknown} & \mbox{otherwise}
  \end{cases}
\end{equation*}
The same kind of case divisions are straight froward for 
$i$ replaced with $i+1$, denoted by $A_{i+1}$ etc.

We consider the following cases:
\begin{description}
\item[Case $+$] : Case $A'_i$ and Case $A'_{i+1}$ occur, 
\item[Case $-$] : Case $B'_i$ and Case $B'_{i+1}$ occur, 
\item[Case ${}^t-$] : Case $C'_i$ and Case $C'_{i+1}$ occur, 
\end{description}
and orthogonally the cases:
\begin{description}
\item[Case $1+2^a$] : Case $A_i$ and Case $A_{i+1}$ occur, 
\item[Case $1-2^a$] : Case $B_i$ and Case $B_{i+1}$ occur, 
\item[Case $2^a-1$] : Case $C_i$ and Case $C_{i+1}$ occur. 
\end{description}

We assume that both one of the cases $+$, $-$, ${}^t-$ and 
one of the cases $1+2^a$, $1-2^a$, $2^a-1$ occur.
For $n=3$ an approximated probability for this is:
\[
\left(\left(\frac{27}{64}\right)^2 \times 3\right)^2 \approx 0.285087 \cdots. 
\]

For example, assume that the case $-$ and $1-2^a$ occur. 
Then, we have 
\begin{align*}
  z &\approx (s_{i+1}-s_i(1-2^a)) + (s_{i+2}-s_{i+1}(1-2^a)) \\
  &= (s_{i+1}+s_{i+2}) + (2^a-1) (s_i+s_{i+1}) \\
  &= y + (2^a-1)x, 
\end{align*}
where we omit modulo $2^{64}$.

\begin{table}[htpb]
  \caption{$z$ from $x$ and $y$ by case division}
  \label{casetable}
  \begin{tabular}{|c|c|c|c|}
    \hline
    Case\textbackslash Case & $1+2^a$ & $1-2^a$ & $2^a-1$ \\ \hline
    $+$ & $(1+2^a)x+y$ & $(1-2^a)x+y$ & $(2^a-1)x+y$ \\ \hline 
    $-$ & $-(1+2^a)x+y$ & $(2^a-1)x+y$ & $(1-2^a)x+y$ \\ \hline 
    ${}^t-$ & $(1+2^a)x-y$ & $(1-2^a)x-y$ & $(2^a-1)x-y$ \\ \hline 
  \end{tabular}
\end{table}

A straight forward computation by case division (Table \ref{casetable}) 
gives that with non negligible probability one of 
\[
z \approx \pm(1+2^a) x \pm y, \quad z \approx \pm(2^a-1)x \pm y.
\]
hold. 
This shows that the consecutive three outputs $(x,y,z)$ by xorshift128+
tend to lie on eight planes, which give an explanation 
on Figure \ref{figure_magnified} \cite{arXiv1907.03251}. 
We compare these planes with the outputs of xorshift128+.

\begin{figure}
  \includegraphics[scale=0.4]{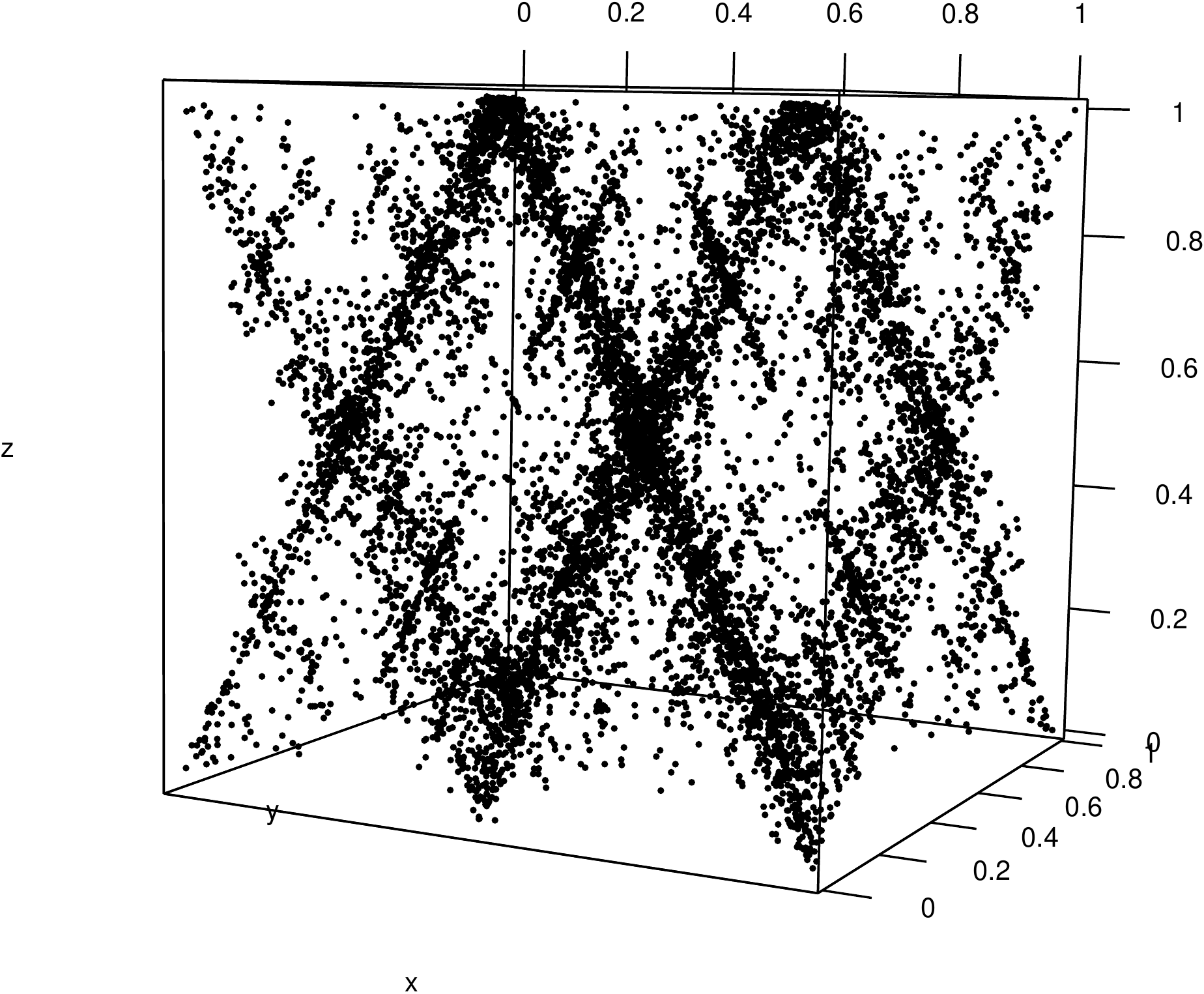}
  \caption{3D plots by xorshift128+: $x$-axis magnified by a factor of 
    $2^{22}$}
  \label{figure_magnified}  
\end{figure}

Figure 2 describes four planes 
\[
z = \pm (2^{23}+1) x \pm y \bmod{1}
\]
with restriction $0 \leq x \leq 1/2^{23}$, $0 \leq y \leq 1$. 
The $x$-axis is magnified with the factor $2^{23}$.
The other four planes with coefficient $2^{23}-1$ are 
very close to those for $2^{23}+1$, and so omitted.
Each plane has two connected components in this region.
Figure 3 shows the union of these four planes.

\begin{figure}[h]
  \begin{minipage}[b]{0.4\linewidth}
    \centering
    \includegraphics[scale=0.4]{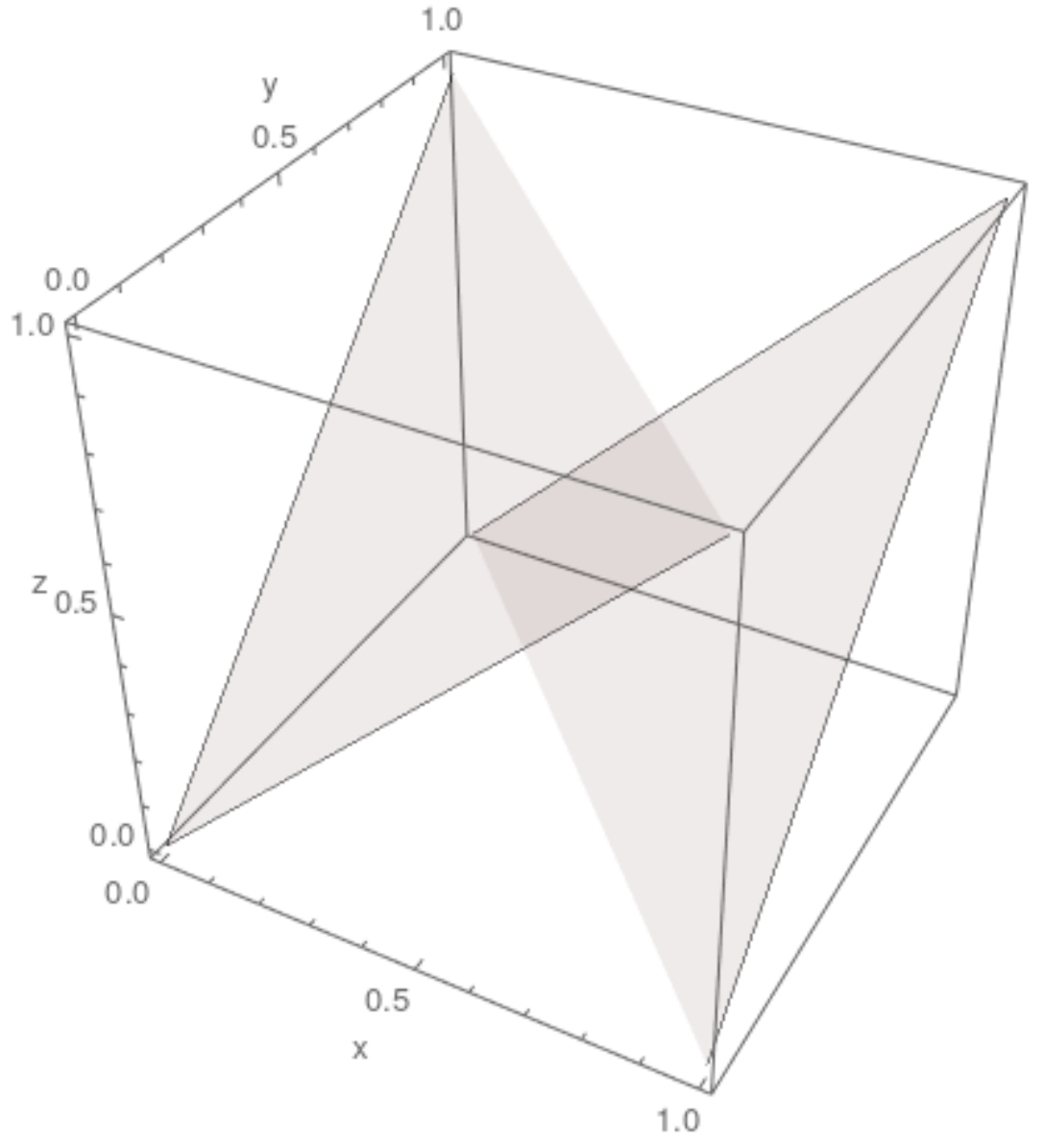}
    \subcaption{}
  \end{minipage}
  \begin{minipage}[b]{0.4\linewidth}
    \centering
    \includegraphics[scale=0.4]{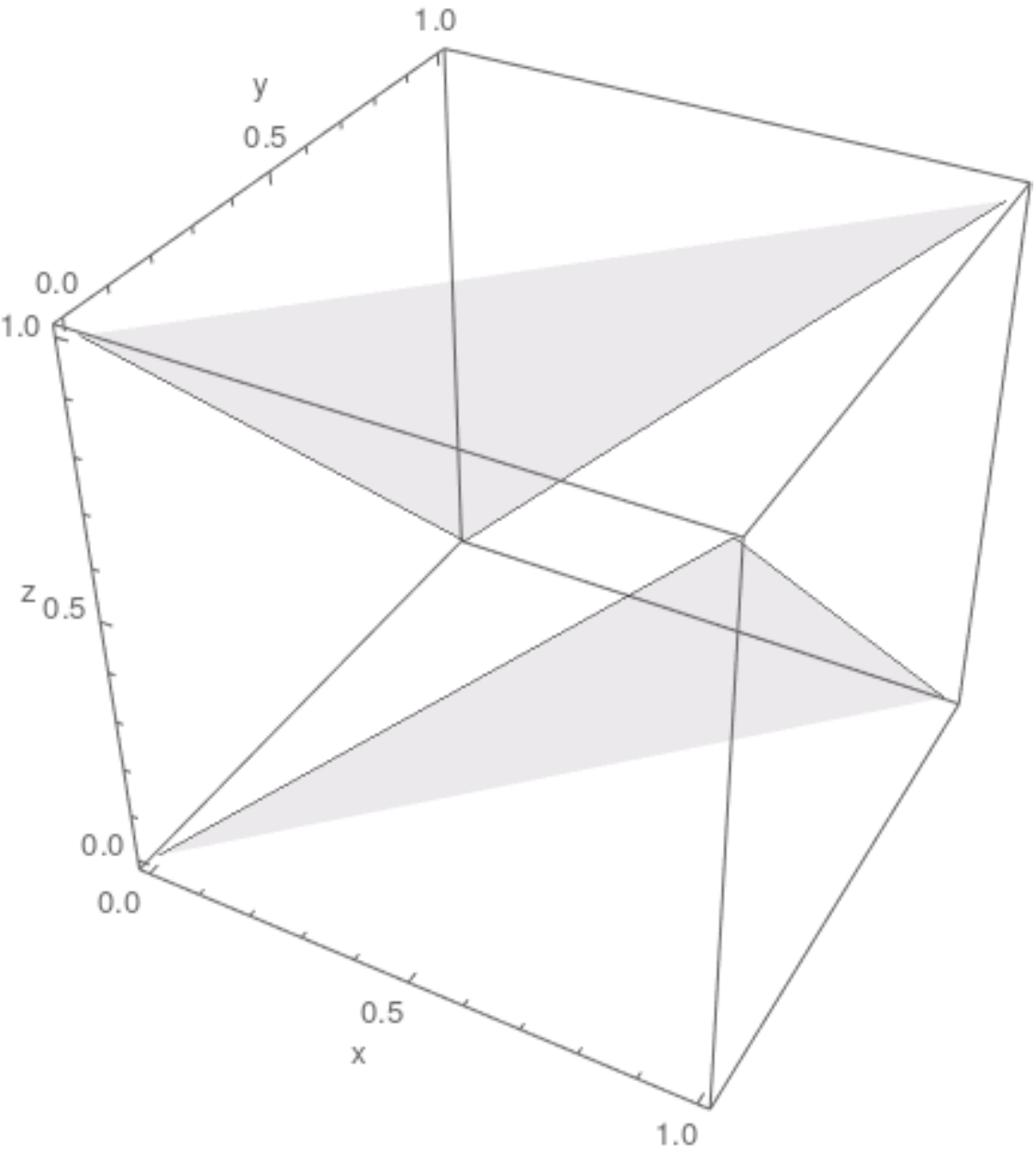}
    \subcaption{}
  \end{minipage}
  \begin{minipage}[b]{0.4\linewidth}
    \centering
    \includegraphics[scale=0.4]{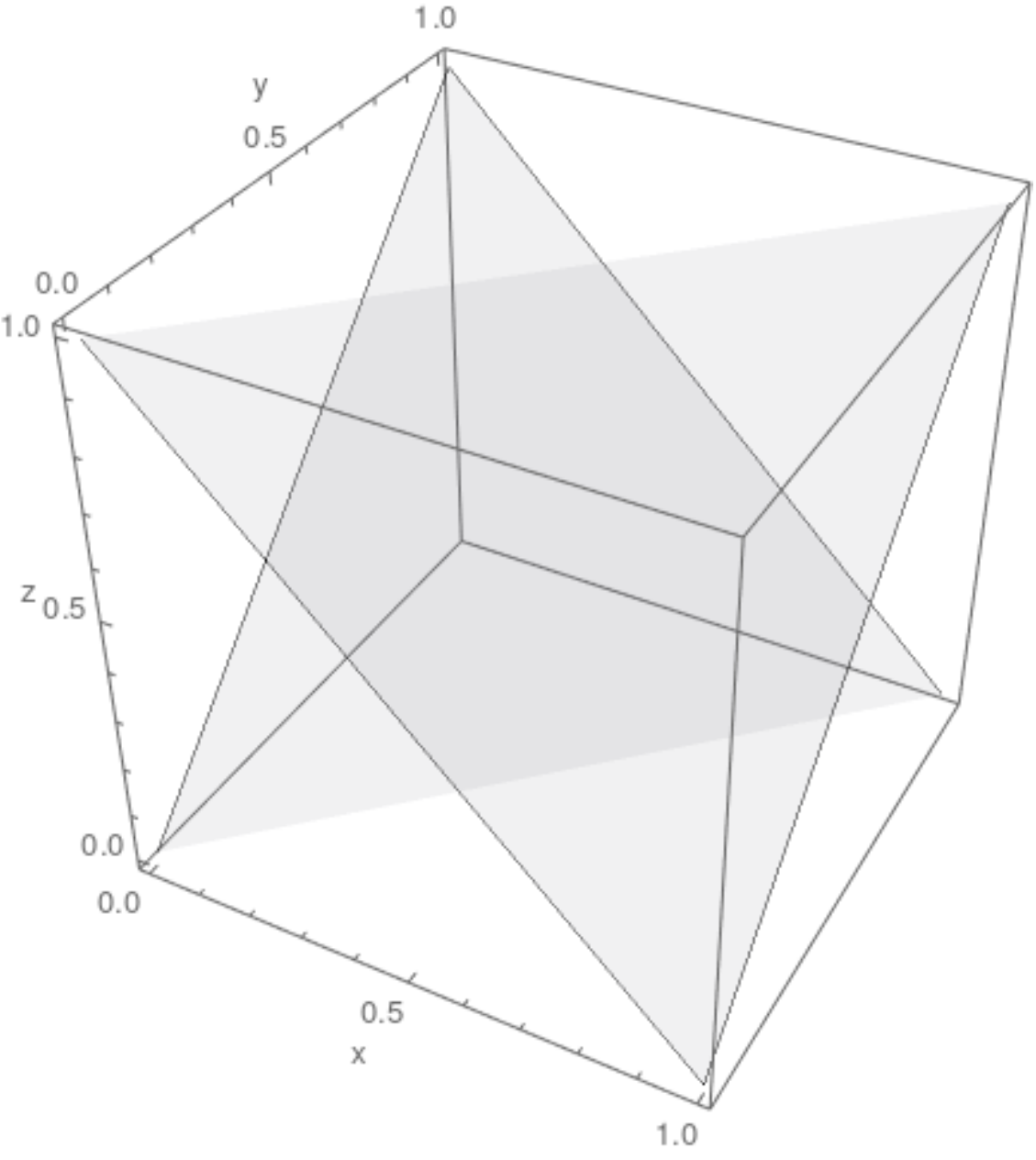}
    \subcaption{}
  \end{minipage}
  \begin{minipage}[b]{0.4\linewidth}
    \centering
    \includegraphics[scale=0.4]{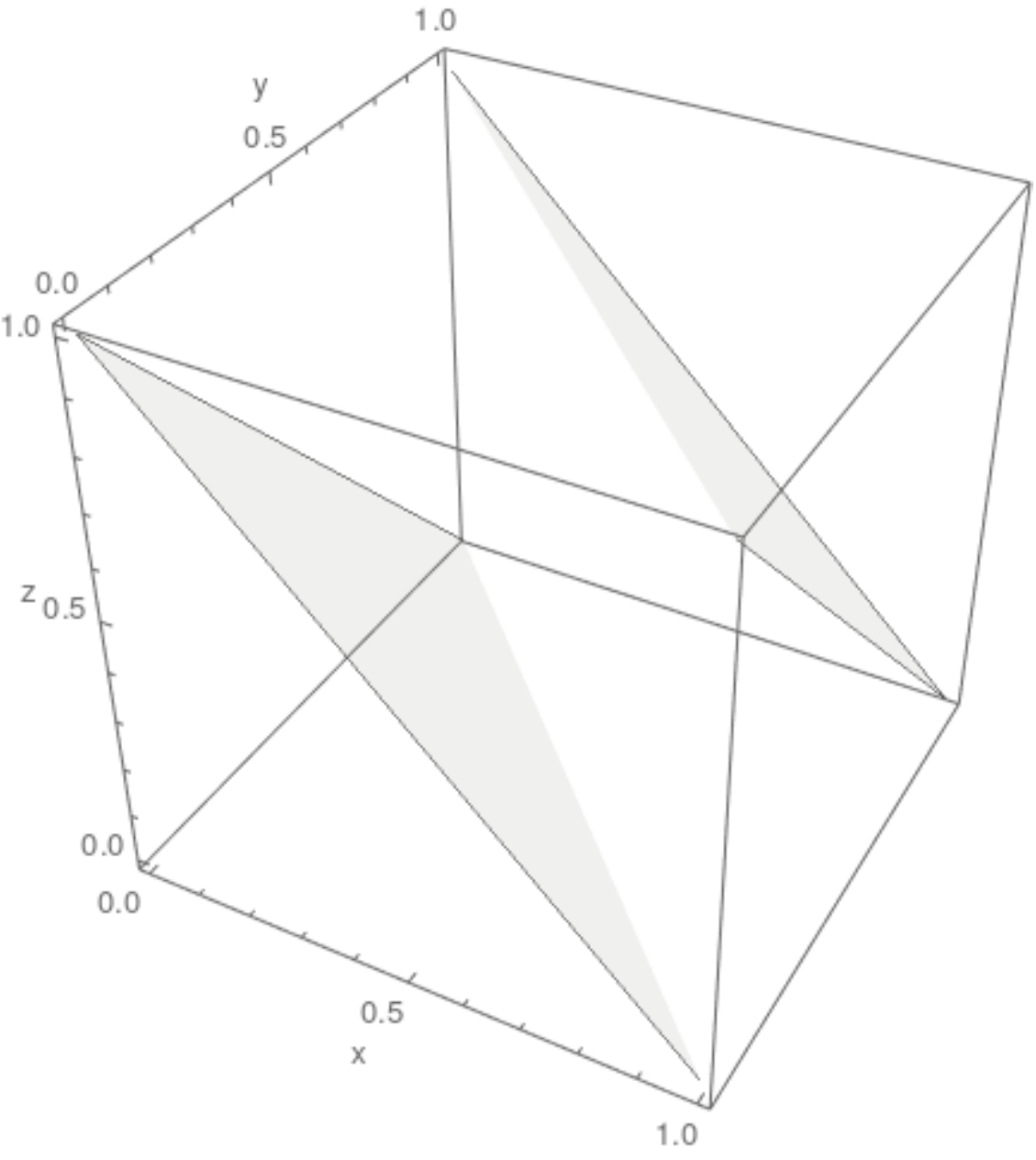}
    \subcaption{}
  \end{minipage}
\caption{Pictures of four planes: (a): $z=(2^{23}+1)x+y \bmod{1}$ , 
(b): $z=(2^{23}+1)x-y \bmod{1}$, 
(c): $z=-(2^{23}+1)x+y \bmod{1}$, 
(d): $z=-(2^{23}+1)x-y \bmod{1}$ 
}
\end{figure}

\begin{figure}[htbp]
  \includegraphics[scale=0.4]{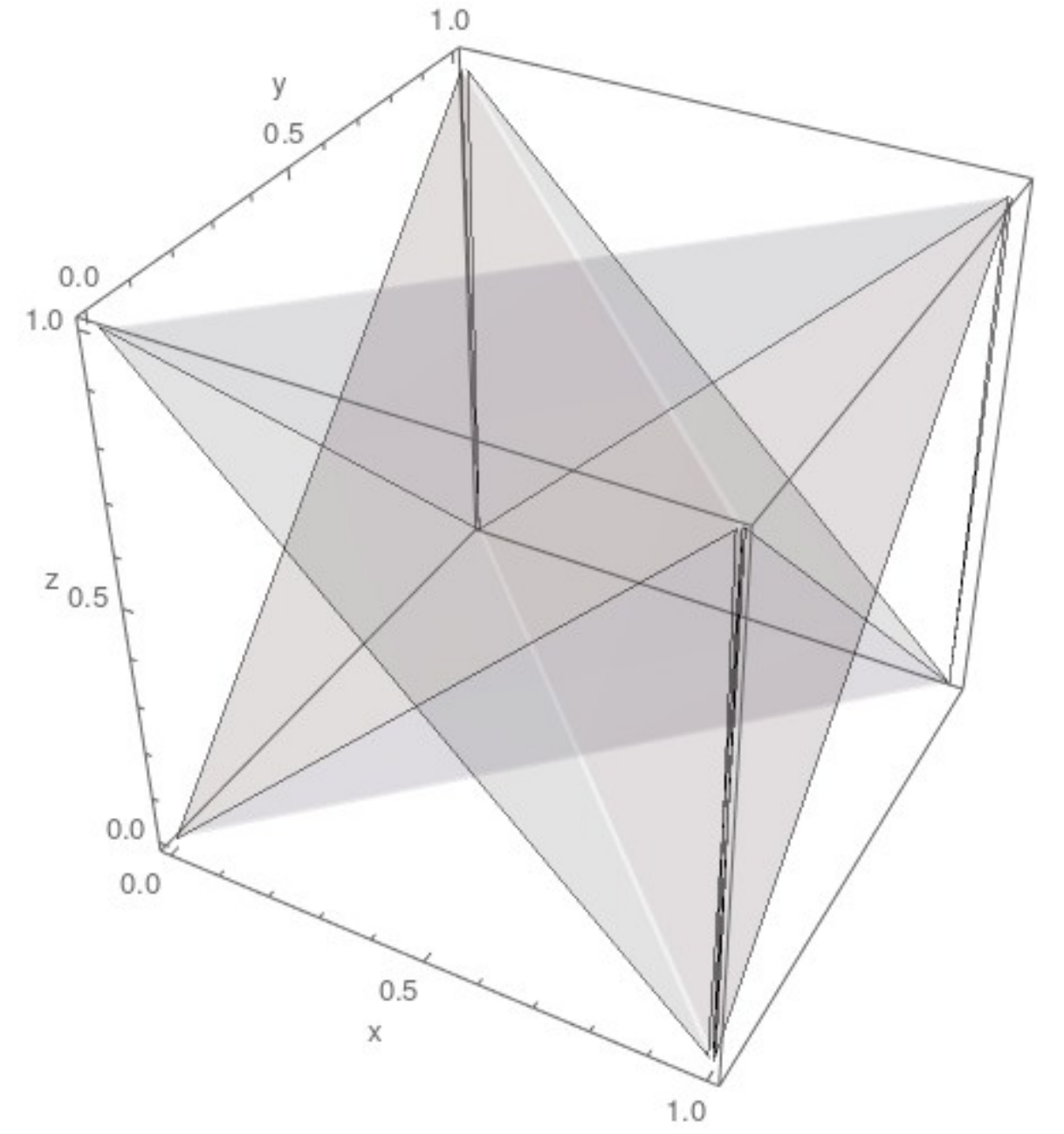}
  \caption{The union of four planes}
\end{figure}

Figure 4 shows the outputs of xorshift128+ with parameter $(a,b,c)=(23,17,26)$. 
Let $(x,y,z)$ be the consecutive outputs in $[0, 1)^3$. 
We only pick up those with $x \leq 1/2^{23}$, and plot $(2^{23}x, y, z)$.
We repeat this until we obtain $10000$ points. 

\begin{figure}[htbp]
  \includegraphics[scale=0.4]{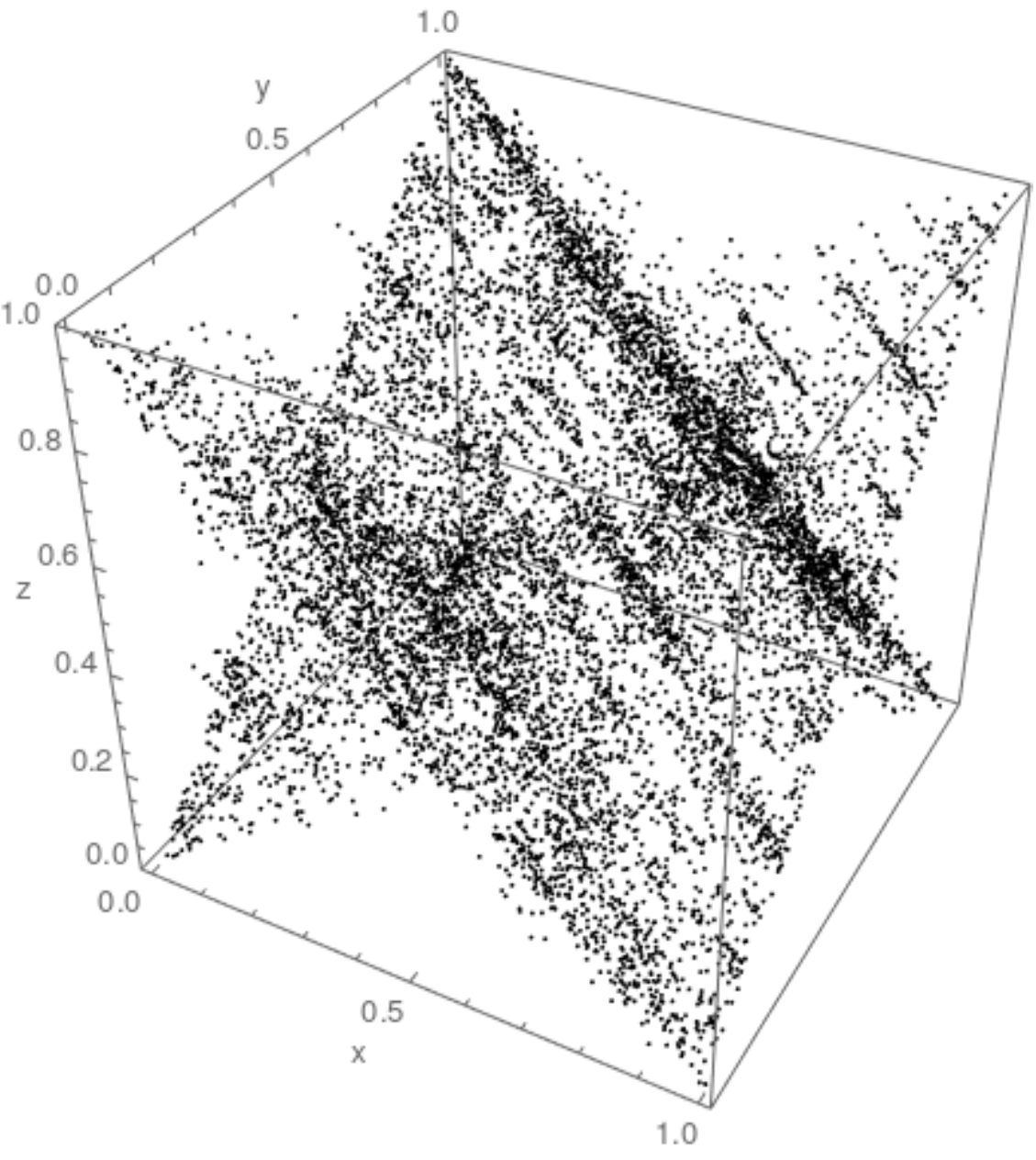}
  \caption{3D plots by xorshift128+: $x$-axis magnified by a factor of $2^{23}$}
\end{figure}

Figure 5 contains both the four planes (Figure 3) and 
the outputs of xorshift128+ (Figure 4). 
This coincidence justifies the approximated analysis done in this section.

\begin{figure}[htbp]
  \includegraphics[scale=0.4]{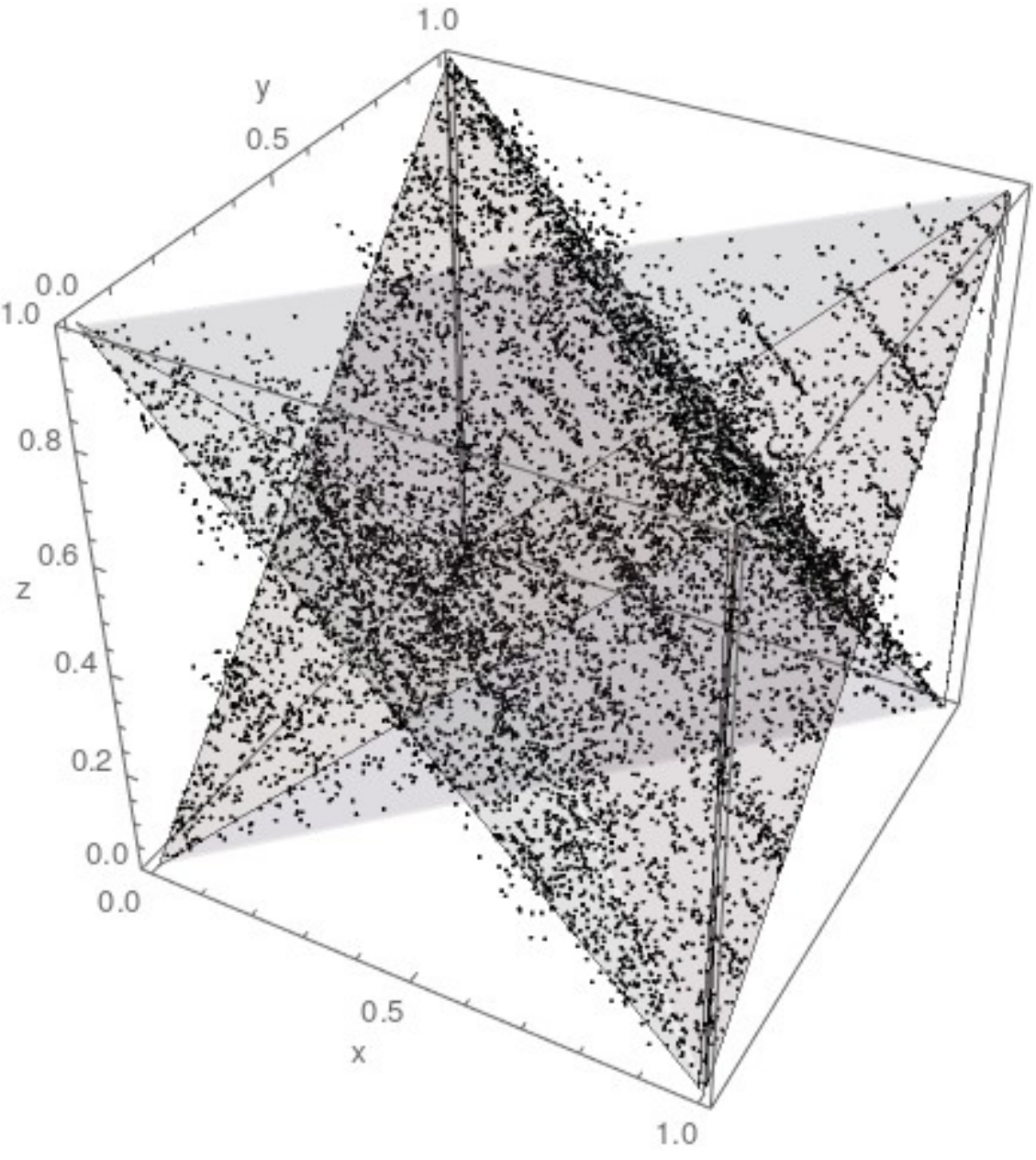}
  \caption{The union of Figure 3 and Figure 4}
\end{figure}

\section{Conclusion}
G. Marsaglia said ``random numbers fall mainly in the planes'' 
\cite{Marsaglia25}.
Experiments show that Vigna's xorshift128+ have this property. 
An analysis of this phenomenon based on approximation of xor by 
arithmetic addition and subtraction is discussed. 

\begin{acks}
This work was supported by JSPS KAKENHI Grant Numbers 26310211, 
JP17K14234 and JP18K03213.
\end{acks}

\bibliographystyle{ACM-Reference-Format}
\bibliography{haramoto01}










\end{document}